\documentclass{amsart}
\usepackage[margin=3cm]{geometry}

\usepackage{graphicx} %

\usepackage{algorithm}
\usepackage{algpseudocode}
\usepackage{xcolor} 
\usepackage{mathrsfs}
\usepackage{multirow}

\usepackage{cleveref} 

\newtheorem{thm}{Theorem}[section]

\newtheorem{lem}[thm]{Lemma}
\newtheorem{cor}[thm]{Corollary}

\theoremstyle{definition}
\newtheorem{defn}[thm]{Definition}

\newtheorem{eg}[thm]{Example}

\newcommand{\C}{\mathcal C}
\newcommand{\cL}{\mathcal L}
\newcommand{\LL}{\mathcal L}

\title{Phylogenetic network classes through the lens of expanding covers}
\author[Andrew Francis, Daniele Marchei, and Mike Steel]{Andrew Francis${}^1$, Daniele Marchei${}^{1,2}$, and Mike Steel$^3$ }
\address{${}^1$ Centre for Research in Mathematics and Data Science, Western Sydney University, Australia}
\address{${}^2$ Computer Science, University of Camerino, Camerino, Italy}
\address{${}^3$ Biomathematics Research Centre, University of Canterbury, New Zealand}
\date{\today}

\begin{document}

\begin{abstract}
    It was recently shown that a large class of phylogenetic networks, the `labellable' networks, is in bijection with the set of `expanding' covers of finite sets. In this paper, we show how several prominent classes of phylogenetic networks can be characterised purely in terms of properties of their associated covers. These classes include the tree-based, tree-child, orchard, tree-sibling, and normal networks.
\end{abstract}


\maketitle

\section{Introduction}

Phylogenetic networks can provide more complete representations of evolutionary relationships among species than possible with a simple phylogenetic tree~\cite{bap13, hus19}. Although a single tree can accurately show ancestral speciation events (splitting of lineages), it cannot display reticulate evolution (where the flow of genomic information follows the merging of ancestral lineages).   Well-known reticulate processes in biology include hybridization, horizontal gene transfer, recombination, and endosymbiosis, in both the recent and distant past.   By contrast, rooted phylogenetic networks can explicitly and simultaneously display both speciation and reticulate evolution.  As a result, the mathematical and algorithmic investigation of phylogenetic networks has become a highly active field over the last $\sim$15 years, and numerous classes of networks have been defined and studied \cite{kong2022classes}.

In this paper, we show how a recently introduced correspondence for a large class of phylogenetic networks (the \emph{labellable} networks~\cite{francis2023labellable}) can be used to characterise a number of widely used other classes of network.  Classes of network have been introduced for a variety of reasons, but usually in order to capture some feature that seems biologically important, or because they are mathematically convenient.  Their definitions typically involve constraints on their structures as graphs.  For instance, tree-child networks are those for which no vertex has only reticulations as its children, whereas tree-based networks are those that can be constructed from a base tree by adding additional edges between the tree edges.  

The class of labellable networks contains many commonly studied classes.  They have been shown to correspond to a set of covers of finite sets that satisfy a property called ``expanding''.  We explore features of covers arising from networks, and characterise many of the familiar classes in terms of properties of their associated covers.  It is to be hoped that encoding network properties in the properties of sets of sets will enable some new directions to be pursued in studying phylogenetic networks.

This paper aims to demonstrate how this encoding of labellable networks into covers may be of broad use in the classification of network classes.   Different classes of networks are defined in different ways, and it can be difficult to present a clear hierarchy (there have been several visual attempts, for instance~\cite[Fig.12]{kong2022classes} and~\cite[Fig.6]{francis2023labellable}). Being able to characterise different network classes by the properties of their covers gives a unified framework for defining networks, in the sense that one may add or remove axioms depending on the class of networks one wants to describe. In that sense, moving from one class to another may be just a matter of changing the axioms, providing a potentially useful lens for visualizing the relationships among classes. 

We begin by defining what we mean by a phylogenetic network, recalling the key results linking labellable networks with expanding covers (from~\cite{francis2023labellable}), in \Cref{s:prelim}.  We give some general properties of covers arising from networks, before characterising the classes of tree-based labellable networks (\Cref{s:tbn}), then tree-child networks (\Cref{s:tree-child}), normal networks (\Cref{s:normal}), tree-sibling networks (\Cref{sec:treesibling}),  and orchard networks (\Cref{s:orchard}). These are some of the more widely seen classes, and they are amenable to being described in terms of covers.  We also demonstrate how the language of covers can allow one to define new classes of network by changing the constraints on the covers:  one small change to the constraints defines a new class we call `spinal' networks, that have an interesting structure (\Cref{s:new.classes}).  We finish by discussing some open questions and opportunities for further development.

\section{Preliminaries}\label{s:prelim}

A \emph{phylogenetic network} on $n$ leaves is a directed acyclic graph with a single vertex of in-degree zero, called the root, and $n$ vertices of in-degree 1 and out-degree zero, labelled by $[n]:=\{1,\dots,n\}$. Note that this includes the possibility of vertices that have in-degree and out-degree both equal to 1, or both strictly greater than 1; such vertices are called \emph{degenerate}.  If $N$ has any degenerate vertices, it is said to be a \emph{degenerate} network; otherwise, it is \emph{non-degenerate}.  

If every vertex has in-degree and out-degree at most 2, then the network is said to be \emph{binary}.  If $N$ is non-degenerate and binary, then all vertices other than the leaves and root have total degree 3.

Vertices in a network that have in-degree 1 are called \emph{tree vertices}, and those with in-degree greater than 1 are called \emph{reticulate vertices}, or \emph{reticulations}.  We will typically use $k$ to denote the number of reticulations in a network, and $m$ to denote the number of non-root vertices in total.

A \emph{labellable} phylogenetic network is one whose vertices can be deterministically labelled according to an algorithm that generalises one for trees (the algorithm for trees is due to Erd\H{o}s and Sz\'ekely~\cite{erdos1989applications})~\cite{francis2023labellable}.  Such networks are characterized topologically by the property that the map from non-leaf vertices to their sets of children is one-to-one~\cite[Thm.3.3]{francis2023labellable}.  

A \emph{partition} of a finite set $A$ is a set of non-empty, pairwise disjoint subsets of $A$ whose union is $A$.  A \emph{cover} of a finite set $A$ is a set of non-empty subsets of $A$ whose union is $A$.  The cardinality $|\C|$ of a cover $\C$ is the number of sets it contains.  We use $||\C||$ to denote the number of distinct elements in the sets in $\C$, that is, $||\C||:=|\bigcup_{C_i\in\C}C_i|$.

Recall the definition from~\cite{francis2023labellable}:
\begin{defn}\label{d:expanding.cover}
A cover $\C$ of $[m]$ is \emph{expanding} if, for $n=m-|\C|+1$, it satisfies:
\begin{enumerate}
	\item No element of $[n]$ appears more than once, and
	\item For $i=1,\dots,|\C|$, the cover contains at least $i$ subsets of $[n+i-1]$.
\end{enumerate}
\end{defn}

\begin{thm}\cite[Thm. 4.4]{francis2023labellable}
    The class of labellable phylogenetic networks is in bijection with the collection of expanding covers of finite sets.
\end{thm}

The map from a labellable phylogenetic network to its expanding cover takes each non-leaf vertex to the set of labels of its children.  That is, sets in the cover are sets of labels of sibling vertices sharing a parent.  The map from an expanding cover $\C$ to a labellable network is a constructive map that first establishes the number of leaves in the network via the following formula~\cite[Lemma~4.1]{francis2023labellable}: 
\[n=||\C||-|\C|+1.\]
The construction of the network then begins with $n$ isolated leaf vertices, and adds parent vertices to sets of vertices present in the growing network, and lexicographically minimal of those in $\C$.  The expanding conditions ensure that there is always such a set, and that the map is well-defined. For examples of this construction the reader is referred to~\cite{francis2023labellable}.

\begin{figure}[ht]
    \includegraphics[]{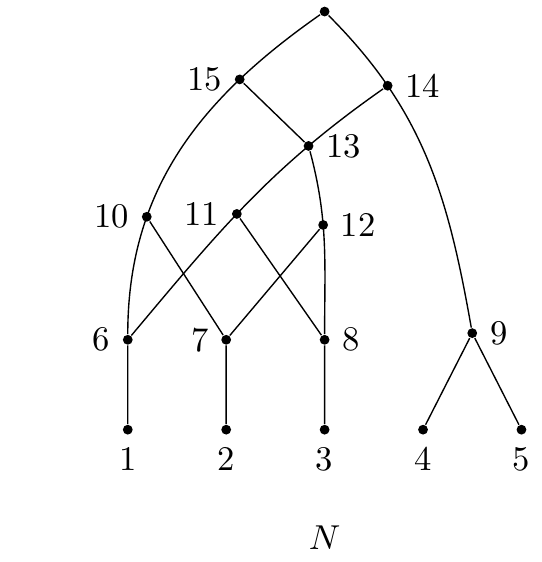}
    \caption{A labellable phylogenetic network $N$ with cover $1\mid 2\mid 3\mid 4,5\mid 6,8\mid 6,7\mid 7,8\mid 11,12\mid 9,13\mid 10,13\mid 14,15$.}
    \label{f:network}
\end{figure}

While the condition for a cover to be expanding may seem artificial, and it certainly restricts from the collection of all covers of a set, it can be seen as a natural extension of the notion of partitions.  In particular, it turns out that all partitions are expanding covers.

\begin{lem}\label{l:partitions.expand}
Every set partition is an expanding cover.
\end{lem}
\begin{proof}
Let $\pi$ be a partition of $[m]$ with $\ell=|\pi|$ blocks, and set $n=m-\ell+1$.
Two conditions define an expanding cover.  The first is that elements of $\{1,\dots,n\}$ are not repeated in $\pi$, which is satisfied by virtue of $\pi$ being a partition.  The second is that for each $i=1,\dots,\ell$, $\pi$ contains at least $i$ subsets of $[n+i-1]$, and we prove this by induction on $i$.

First, consider the base case $i=1$. We need to show that there is at least one set in $\pi$ that is a subset of $[n]$. There are $\ell=m-n+1$ pairwise disjoint subsets of $[m]$ in $\pi$, and there are $m-n$ integers in $[m]$ that are not in $[n]$.  Therefore, there must be at least one set in $\pi$ that does not contain an element of $\{n+1,\dots,m\}$ and is thus in $[n]$, as required.

Suppose that for $i=k$, $\pi$ contains at least $k$ subsets of $[n+k-1]$.  We would like to show that $\pi$ contains at least $k+1$ subsets of $[n+(k+1)-1]=[n+k]$.  The proof proceeds in the same manner as the case of $i=1$.

First remove $k$ subsets of $[n+k-1]$ from $\pi$, so that $\pi$ has $\ell-k$ sets remaining.  We need to show at least one remaining set is entirely contained within $[n+k]$.  There are $m-(n+k)$ integers in $\pi$ that are \emph{not} in $[n+k]$, and 
$\ell-k=(m-n+1)-k=m-(n+k)+1$
sets are available.  Therefore, at least one must not contain any element outside $[n+k]$, as required.
\end{proof}

Since all set partitions are expanding covers, we can ask what sort of networks have partitions as their covers.  A partition has a single occurrence of each integer, which means that each vertex of the network (each label) has a single set of siblings.  In other words, the network has no reticulations, and thus is a tree.  This correspondence of trees with partitions allows trees with degenerate vertices (i.e., vertices with in-degree and out-degree 1).  In this way, the correspondence for partitions is closer to the result of Erd\H os and Sz\'ekely~\cite{erdos1989applications} than the non-degenerate framework that has partitions in bijection with phylogenetic forests in~\cite{francis2022brauer}. 

The lexicographic order on sets (given by $A\prec B$ if $A\subset B$ or $\min(A\setminus B)<\min(B\setminus A)$) that helps determine the labelling sequence is not always the ordering of sets used to label the internal vertices of the network; that sequence is given by the \emph{labelling order}, which is defined as follows~\cite[Section 4]{francis2023labellable}:

\begin{defn}\label{d:labelling.order}
    The \emph{labelling order} for an expanding cover $\C$ is determined by the following procedure. 
    \begin{enumerate}
        \item For $i=1,\dots,|\C|$, 
        \begin{enumerate}
            \item Set $C_i$ to be the minimal set in $(\C, \prec)$ contained in $[n + i - 1]$; and
            \item Redefine $\C=\C\setminus\{C_i\}$.
        \end{enumerate}
        \item Output the sequence $C_1,\dots,C_{|\C|}$.
    \end{enumerate}
\end{defn}

This order is necessary to establish conditions on a cover that give non-degenerate networks, for instance, and we will use it later in the present paper to describe \emph{normal networks} (in~\Cref{s:normal}) and \emph{orchard networks} (\Cref{s:orchard}).

Given a cover in labelling order, we can label every subset in position $1 \leq i < |\C|$ by $i + n$, whereas the last subset is labelled $\rho$ for the \emph{root}. In this way, the label for each subset corresponds to the label of its parent in the corresponding labellable network.

For example, the labelling order for the network shown in~\Cref{f:network} is 
\[1\mid 2\mid 3\mid 4,5\mid 6,7\mid 6,8\mid 7,8\mid 11,12\mid 9,13\mid 10,13\mid 14,15.\]
The first set gives rise to the vertex label $n+1=6$, the second gives rise to $7$, and so on. We can represent this more explicitly as follows, adding $\rho$ to denote the root:
\[\{1\}_6, \{2\}_7, \{3\}_8, \{4,5\}_9, \{6,7\}_{10}, \{6,8\}_{11}, \{7,8\}_{12}, \{11,12\}_{13}, \{9,13\}_{14}, \{10,13\}_{15},\{14,15\}_\rho.\]

\subsection{Features of vertices in networks and their covers' properties}
\label{s:cover.properties}

Many features of vertices in networks have direct translations into the language of covers, and we present some of them in~\Cref{tab:dictionary}.  The first two lines of the table are clear: non-root vertices on a network are labelled by the labelling algorithm and those labels appear as integers in $[m]$, and the leaves are labelled by integers in $[n]$.  The other lines of the table can be justified as follows. 

A tree vertex in a network is a vertex with in-degree 1, which means it has only one parent and, therefore, is in only one set of sibling vertices.  This set of sibling vertices could have any size greater than or equal to one, but it is only a single set.  A reticulation vertex, on the other hand, has strictly more than one parent, and thus has two or more sets of siblings. No two vertices in a labellable network have the same set of children~\cite[Thm~3.3]{francis2023labellable}, so the label of a reticulation vertex will appear in at least two sets in the cover.
The other translations in~\Cref{tab:dictionary} follow immediately.

Throughout this paper, we will add additional translations to the table, with a summary table given in the Discussion.

\begin{table}[ht]
    \centering
    \begin{tabular}{lp{10cm}}
         \hline
         \textbf{Network} & \textbf{Cover}  \\ 
         \hline
         Non-root vertex & An integer in $[m]$ \\
         Leaf & An integer in $[n]$ \\
         Tree vertex & An integer contained in just one subset \\
         Reticulation vertex & An integer contained in more than one subset \\
         In-degree of $x$ & The number of subsets that contain $x$ \\
         Out-degree of $x$ & Size of the subset with label $x$ in the labelling order \\
         Parents of $x$ & All the subsets that contain $x$ \\
         Siblings of $x$ & All the other integers contained in the subsets that contain $x$ \\
         Children of $x$ & The subset with label $x$ in the labelling order\\
         \hline\\
    \end{tabular}
    \caption{A translation of features of vertices in a network with $n$ leaves and $m$ non-root vertices into features of the corresponding expanding cover.  
}
    \label{tab:dictionary}
\end{table}

\section{Tree-based networks}\label{s:tbn}

A phylogenetic network is \emph{tree-based} if it has a spanning tree whose leaves are those of the network~\cite{francis2015phylogenetic}.  Such a spanning tree is called a \emph{base tree} for the network. Typically, a tree-based network can have many base trees.  A similar notion that we will discuss is that of a \emph{support tree} for a network.  A support tree is a base tree but with additional degree 2 vertices where additional arcs are joined to complete the network.  That is, the set of vertices in the support tree and the network are identical.

Unlike the other classes that we consider in the coming sections, not all tree-based networks are labellable, but neither are all labellable networks tree-based~\cite{francis2023labellable}.  There is thus a non-trivial intersection of the two classes, and this intersection contains many other classes, including orchard, tree-child, and normal networks~\cite{francis2023labellable}.  In the binary case, the tree-based networks that are labellable can be characterised in terms of their structural properties, as those for which no two reticulate vertices have the same sets of parents~\cite[Thm.~6.3]{francis2023labellable}.  In this section, we provide a new characterisation of the tree-based labellable networks in terms of their covers, and the existence of an ``embedded'' partition, in~\Cref{t:treebased.embedded.partitions}. 

We say that a partition $\pi$ \emph{embeds} in $\C$ if there is a one-to-one map from $\pi$ to $\C$ that maps each set $A$ in $\pi$ to a set $A'$ in $\C$ so that $A\subseteq A'$.  A partition $\pi$ \emph{fully embeds} in a cover $\C$ if $\pi$ embeds in $\C$ and $|\pi|=|\C|$.

Recall from~\Cref{s:prelim} that every partition of $[m]$ is an expanding cover.  It is straightforward to see that every expanding cover has a partition that embeds into it, as follows.

\begin{lem}\label{l:expanding.partition.embedding}
Every expanding cover of $[m]$ has an embedded partition of $[m]$. 
\end{lem}

\begin{proof}
If all repeats of integers are deleted, so that there is one occurrence of each integer, then the result is a partition of $[m]$.  
\end{proof}

Any partition obtained in this way will be expanding, according to  Lemma~\ref{l:partitions.expand}.  Note, however, that each such partition may not have the same number of sets as the cover, and therefore may be expanding for a different value of $n$.

The notion of embedding a partition into a cover turns out to help characterise tree-based networks.

\begin{thm}\label{t:treebased.embedded.partitions}
An expanding cover $\C$ of $[m]$ corresponds to a tree-based network if and only if it has a fully embedded partition $\pi$ of $[m]$. 
\end{thm}

\begin{proof}
Suppose $N$ is a tree-based network with expanding cover $\C$ of $[m]$.  We will show that $\C$ has an embedded partition with length $|\C|$.

Label the vertices of $N$ according to the labelling algorithm.  This labelling gives rise to the expanding cover whose sets are the children of non-leaf vertices in $N$. Choose a support tree $T$ for $N$, keeping the labels of the vertices from $N$.  The labels of vertices in $T$ are thus precisely $[m]$. Note that all vertices of $N$ are present in $T$, but that each non-root vertex in $T$ has in-degree 1.  The set of children of each vertex in $T$ is a subset of the set of children for the corresponding vertex in $N$.  

Construct the cover for $T$ using the inherited labelling of vertices, forming sets of labels of vertices that are the children of the same non-leaf vertex. Each set thus formed is a subset of one of the sets in the cover for $N$, because the children of vertex $i$ in $N$ are a subset of the children of vertex $i$ in $T$.  Each set is non-empty because the only leaves in the base tree are those of $N$.  The cover for $T$ contains no repeated integers because $T$ is a tree and there are no vertices with in-degree greater than 1.  Thus, the cover for $T$ with the labelling inherited from $N$ is a partition of $[m]$ of length $|\C|$, as desired.

Note that the labels on the vertices in $T$ are those inherited from $N$.  They are not the same as the labels that would be put on vertices by the labelling algorithm applied to $T$. Thus the partition obtained from $T$ is not the same as the partition that would be obtained by labelling $T$ directly.

For the reverse direction, suppose that the expanding cover $\C$ has an embedded partition $\pi$ with length $|\C|$.  We will show that the corresponding network is tree-based.

Let $N$ be the network constructed by using $\C$.  The partition $\pi$ embeds in $\C$, so there is a one-to-one map from $\pi$ to $\C$ that maps each set $A$ in $\pi$ to a set $A'$ in $\C$ such that $A\subseteq A'$.  The sets in $\C$ correspond to vertices in $N$ and give the set of children of each vertex.  For each non-leaf vertex in $N$,  $A'\in\C$ labels its children, and there is a corresponding set $A\in\pi$ that is its pre-image in the embedding of $\pi$ into $\C$, with $A\subseteq A'$.  

For the non-leaf vertex in $N$ with children $A'$, delete the edges in $N$ between it and the vertices labelled by $A'\setminus A$, and repeat this for each non-leaf vertex in $N$.   The resulting network now has vertices whose children are labelled by the sets in $\pi$.  We claim that this resulting network $\hat N$ is a support tree for $N.$  We need to show that $\hat N$ is a spanning tree whose leaves are those of $N$.

First, $\hat N$ contains all vertices of $N$, since only edges were removed.  Second, it is a tree, since no label is repeated in $\pi$ by virtue of it being a partition, and therefore no vertex has more than one parent.  Third, each vertex $v$ that is not a leaf of $N$ has at least one child, since $v$ has a non-empty set of children whose labels are a set in $\pi$ (the length of $\pi$ is $|\C|$), and thus the only leaves of $\hat N$ are those of $N$.  

Thus, $\hat N$ is a support tree for $N$, and so $N$ is tree-based, as required.
\end{proof}

This result gives an alternative way to characterise support trees for a tree-based network, as follows.  

\begin{cor}\label{l:support.trees.partitions}
The set of support trees for a tree-based network $N$ is in bijection with the set of full embeddings of partitions in the expanding cover for $N$.    
\end{cor}
\begin{proof}
As seen in the proof of~\Cref{t:treebased.embedded.partitions}, each support tree for $N$ gives rise to a full embedding of a partition in the cover for $N$.  Conversely, every full embedding of a partition into the cover for $N$ constitutes a choice of parent for each reticulation vertex (any element that appears more than once in the cover), and thus gives a support tree for $N$. 
\end{proof} 

Note that it is possible for a particular partition to embed in more than one way into a cover, and that each such embedding gives a different support tree for the network.

\begin{eg}\label{eg:embedded.partitions}
    \Cref{f:network} shows a network with cover $\C=1\mid 2\mid 3\mid 4,5\mid 6,8\mid 6,7\mid 7,8\mid 11,12\mid 9,13\mid 10,13\mid 14,15$. The embeddings of partitions into $\C$ can be enumerated as follows.  First, consider the elements that appear exactly once in $\C$: $1,2,3,4,5,9,10,11,12,14,15$.  These must appear in the partition where they are in the cover (one appearance means only one possibility), so any embedded partition into $\C$ has form \[1\mid 2\mid 3\mid 4,5\mid {\_},\_\mid \_,\_\mid \_,\_\mid 11,12\mid 9,\_\mid 10,\_\mid 14,15.\]
    Consider then the integer 6, which, in the partition, must be either embedded into the set $\{6,7\}$ or $\{6,8\}$. If the former, then $8$ must embed into the latter; otherwise, the partition would not be a full embedding (we cannot allow empty sets), which forces 7 to embed into the set $\{7,8\}$.  In short, the three sets $6,8\mid 6,7\mid 7,8$ can only have embedded either $6\mid 7\mid 8$ or $8\mid 6\mid 7$.  These amount to the same partition but two distinct embeddings that give different support trees because they correspond to different choices of child for each vertex.  The other choice for embedding a partition involves the placement of $13$, which can either be with 9 or 10.

    Thus, there are four full embeddings of partitions $\pi_i$ into $\C$, as follows:
\begin{center}
    \begin{tabular}{r @{\quad} c @{ $\mid$ } c@{ $\mid$ }c@{ $\mid$ }c@{ $\mid$ }c@{ $\mid$ }c@{ $\mid$ }c@{ $\mid$ }c@{ $\mid$ }c@{ $\mid$ }c@{ $\mid$ }c}
        $\C:$ & 1& 2& 3& 4,5& 6,8& 6,7& 7,8& 11,12& 9,13& 10,13& 14,15 \\
        $\pi_1:$ & 1& 2& 3& 4,5& 6& 7& 8& 11,12& 9,13& 10& 14,15 \\
        $\pi_2:$ & 1& 2& 3& 4,5& 6& 7& 8& 11,12& 9& 10,13& 14,15 \\
        $\pi_3:$ & 1& 2& 3& 4,5& 8& 6& 7& 11,12& 9,13& 10& 14,15 \\
        $\pi_4:$ & 1& 2& 3& 4,5& 8& 6& 7& 11,12& 9& 10,13& 14,15 
    \end{tabular}
\end{center}
The support trees corresponding to these embeddings of partitions are shown in~\Cref{f:support-trees}.

\begin{figure}[ht]
    \centering
    \includegraphics[width=\textwidth]{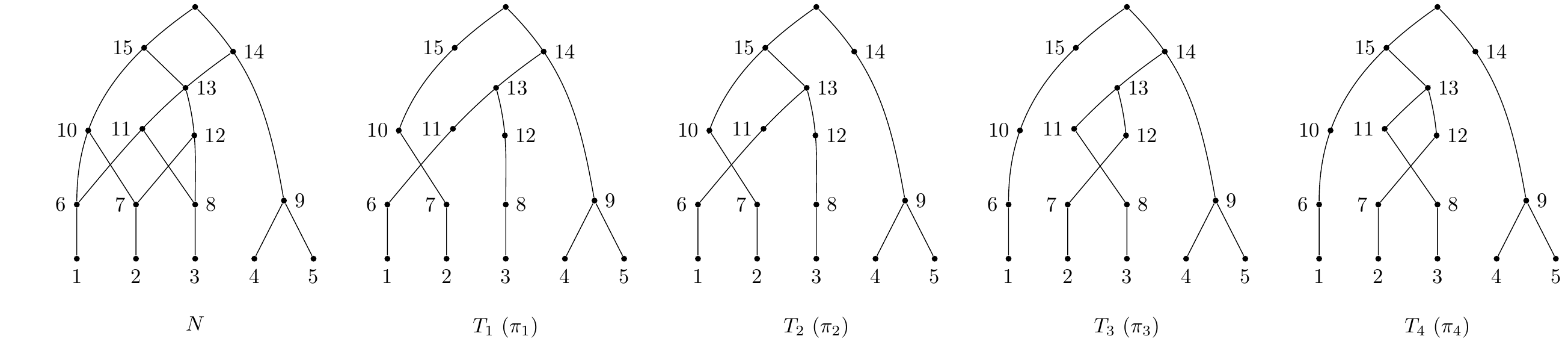}
    \caption{The tree-based network $N$ and the four support trees given by the four embeddings $\pi_1,\dots,\pi_4$, as described in~\Cref{eg:embedded.partitions}. }
    \label{f:support-trees}
\end{figure}
\end{eg}

\begin{table}[ht]
    \centering
    \begin{tabular}{lp{10cm}}
         \hline
         \textbf{Network} & \textbf{Cover}  \\ 
         \hline
         Spanning tree & A partition embedded in $\C$ \\
         Support tree & A full embedding of a partition in $\C$ \\
         \hline\\
    \end{tabular}
    \caption{Translation of concepts arising in tree-based networks.}
    \label{tab:dictionary-spanning.tree}
\end{table}

\subsection{Support trees for a binary tree-based network}

Support trees for binary tree-based networks have been counted in earlier work~\cite{pons2019tree,hayamizu2021structure}, building on an upper bound from~\cite{jetten2015characterising}. Covers provide an alternative and clear approach that replicates these results.

For instance (and without giving details of all the components of the statement):
\begin{thm}[\cite{pons2019tree}, Theorem 8]\label{t:pons.support.trees}
    For a binary tree-based network $N$, the number of support trees is:
    \[2^c\times\prod_{P\in\pi(\mathcal J_N)}\frac{1}{2}(v(P)+1),\]
    where 
    \begin{itemize}
        \item $\mathcal J_N$ is a bipartite graph derived from $N$ with parts given by the set of vertices with a reticulate child, and reticulations without a reticulate parent, 
        \item $c$ is the number of cycle components in $\mathcal J_N$,
        \item $\pi(\mathcal J_N)$ is the set of path components in $\mathcal J_N$ without an omnian terminal vertex, and 
        \item $v(P)$ is the number of vertices in the path component $P$.
    \end{itemize}
\end{thm}
This is an explicit formula based on features of the network, using a representation of key features in the bipartite graph $\mathcal J_N$ in particular.

It was subsequently demonstrated that this formula relied on two key structural elements of the network: the number of ``crowns'' and the lengths of each ``$M$-fence''~\cite[Section 5.3]{hayamizu2021structure}.  These are types of ``zig-zag trails'', which are undirected paths of vertices in the network that alternate between tree and reticulation vertices~\cite{zhang2016tree}. A maximal length zig-zag trail is called a \emph{crown} if it forms a cycle, and is called an \emph{$M$-fence} if the ends of the path are tree vertices.   
Crowns and fences arise naturally when looking at the problem through the lens of covers.  We are able to obtain, by using covers, a formula that is analogous to that of~\Cref{t:pons.support.trees}, as follows.

Suppose $N$ is a binary tree-based network. We allow degenerate vertices with in-degree $2$ as well as out-degree 2.  The cover $\C$ for $N$ then consists of sets of size 1 or 2, and each integer appearing in $\C$ appears either once, if it is a tree vertex (in-degree 1), or twice if it is a reticulation  (in-degree 2). 

We will now describe an algorithm for obtaining an embedded partition (support tree) from $\C$, and this will allow us to count the number of such support trees.

The sets in $\C$ fall into exactly five categories:
\begin{enumerate}
    \item Singletons containing integers appearing once in $\C$, \label{singletons.once}
    \item Singletons containing integers appearing twice in $\C$, \label{singletons.twice}
    \item Pairs containing integers each appearing once in $\C$, \label{pairs.once}
    \item Pairs containing integers each appearing twice in $\C$, and \label{pairs.twice}
    \item Pairs containing one integer appearing once and the other appearing twice in $\C$. \label{pairs.mixed}
\end{enumerate}

Sets that contain elements that appear only once in $\C$ must be fully retained in any embedded partition.  Thus sets from categories \eqref{singletons.once} and \eqref{pairs.once} must be in the embedded partition, and there is no choice.

Because the partition embeds into $\C$, a set containing a singleton $\{a\}$ in $\C$ must also appear in the embedded partition.  Therefore, if $\{a\}$ is in category~\eqref{singletons.twice}, none of the other occurrences of $a$ in other sets in $\C$ can appear in the partition, and we delete them from the sets in the cover.  This will create new sets of size 1, and possibly of category~\eqref{singletons.twice}. We repeat this process until all sets in category \eqref{singletons.twice} are gone, creating a new cover we denote $\C_1$.  Note that $\C_1$ is uniquely determined from $\C$ and embeds into it.  Note also that $\C_1$ does not contain any sets in category \eqref{singletons.twice} above.

This leaves sets from categories \eqref{pairs.twice} and \eqref{pairs.mixed} to deal with.  These sets are connected. If a set is in category~\eqref{pairs.mixed}, then one of its elements appears elsewhere, and it can only be in a set from category \eqref{pairs.mixed} or \eqref{pairs.twice}.  We can thus form sequences of such sets in $\C_1$ by connecting a set from category \eqref{pairs.mixed} with a sequence of sets from category \eqref{pairs.twice} and ending with another set from category~\eqref{pairs.mixed}.  These sequences are uniquely determined by $\C_1$, and every set from category~\eqref{pairs.mixed} is in precisely one sequence of this form.  For example, such sequences are of form 
\begin{equation}\label{eq:mixed-to-mixed}
a_0,a_1\mid a_1,a_2\mid\dots\mid a_{t-1},a_t\mid a_t,a_{t+1},    
\end{equation}
where $a_0$ and $a_{t+1}$ do not appear elsewhere in $\C_1$ (note that $t$ could be 1).  We call such sequences \emph{fences} (they correspond to the $M$-fences defined above).  The notions of crowns and fences for covers are summarized in~\Cref{tab:dictionary-crowns}.

Let $\mathcal F$ denote the set of fences in $N$. For each fence $f$, let $r(f)$ denote the number of repeated integers in $f$, which we call its length. The fence in~\Cref{eq:mixed-to-mixed} has length $r(f)=t$.

A set from category~\eqref{pairs.twice} may be in a sequence such as the one above, or in a sequence of at least three sets from the same category:
\begin{equation}\label{eq:crown}
a_0,a_1\mid a_1,a_2\mid\dots\mid a_{t-1},a_t\mid a_t,a_0,
\end{equation}
where $t\ge 2$.  These correspond precisely to the `{crowns}' of~\cite{hayamizu2021structure}.

For either fences or crowns, we can count the number of selections of unique elements as follows.  

In the case of fences of length $t$ (\Cref{eq:mixed-to-mixed}), the number of choices is simply $t+1$, since there are $t+2$ elements to go into $t+1$ non-empty sets, so one has two elements and the rest have one element. There are $t+1$ choices for the set with two elements.  For example, with the fence $a,b\mid b,c\mid c,d\mid d,e$, we have $t=3$ and the choices are:
\begin{align*}
    &a,b\mid c\mid d\mid e\\ 
    &a\mid b,c\mid d\mid e\\
    &a\mid b\mid c,d\mid e\\
    &a\mid b\mid c\mid d,e.
\end{align*}

In the case of a crown, as in \Cref{eq:crown}, there is only one embedded partition. We have the same number of elements as we have non-empty sets, and so there is only one option for selecting unique elements.  Each element forms a singleton.  For example, in the crown $a,b\mid b,c\mid c,d\mid d,a$, we have only $a\mid b\mid c\mid d$.  However, although there is only one embedded partition, that partition has exactly two distinct embeddings.  We could have:
\begin{align*}
& a\mapsto\{a,b\},\ b\mapsto\{b,c\},\ c\mapsto\{c,d\}\text{ and }d\mapsto\{d,a\},\text{ or}\\
& b\mapsto\{a,b\},\ c\mapsto\{b,c\},\ d\mapsto\{c,d\}\text{ and }a\mapsto\{d,a\}.
\end{align*}

Therefore, we have shown the following result, which is equivalent to~\Cref{t:pons.support.trees}:
\begin{thm}\label{t:counting.support.trees}
Let $N$ be a binary tree-based network with cover $\C$.  The number of embedded partitions in $\C$, and therefore the number of support trees for $N$, is
\[ 
    2^c\times\prod_{f\in\mathcal F}\left(r(f)+1\right)
\]
if $\mathcal F$ is non-empty, and is $2^c$ if $\mathcal F=\emptyset$, where $c$ is the number of crowns in $\C$.
\end{thm}
Note that the number of crowns, $c$, is the same as the number of components referred to in~\Cref{t:pons.support.trees}.

Given a cover $\C$, we can compute the number of crowns and the lengths of fences, and thus the number of embedded partitions, by using \Cref{a:count.embedded.covers}, which uses the definition of `acquaints'.

\begin{defn}\label{d:acquaints}
Set $x\sim y$ if $x=y$ or $x,y$ are siblings, and consider the transitive closure of $\sim$, which is an equivalence relation on the set of vertices of the network. Two vertices in an equivalence relation are said to be {\em acquaints} of each other.
\end{defn}

Acquaints can be defined self-referentially by saying that an \emph{acquaint} of a vertex $x$ is a sibling of $x$ or is a sibling of an acquaint of $x$.
Fences and crowns can be described in terms of acquaints, as follows.

\begin{thm}\label{t:fences.acquaints}\label{t:crowns.acquaints}
    Let $N$ be a binary tree-based network with cover $\C$. Then 
    \begin{enumerate}
        \item $N$ has a fence if and only if there exists a set of acquaints in which exactly two vertices that appear uniquely in $\C$ have one sibling. 
        \item $N$ has a crown if and only if there exists a set of acquaints in which no vertex has one sibling. 
    \end{enumerate}

\end{thm}
\begin{proof}
(1)    For the forward direction, suppose that we have a fence like that in \Cref{tab:dictionary-crowns}.  The integers in the set $\{a_0, a_1, a_2, \dots, a_{t-1},a_t, a_{t+1}\}$ are acquaints, $a_0$ and $a_{t+1}$ have only one sibling ($a_1$ and $a_t$ respectively), and they appear uniquely by assumption.

    Conversely, assume there is a set of acquaints in which exactly two vertices (say $a_i$ and $a_j$) that appear uniquely in $\C$ have one sibling. 
    Since we assume that the network is binary, $a_i$ and $a_j$ appear in only one subset, but they can not be in the same one; otherwise, they would not be acquainted with the other vertices.

    It is also the case that every other vertex will appear in exactly two subsets; otherwise, it would imply an in-degree greater than 2, which is not allowed in a binary network.
    Therefore, we have a set of a type described in \Cref{tab:dictionary-crowns}, and the network has a fence.

(2)
    For the forward direction, suppose we have a crown (as indicated in \Cref{tab:dictionary-crowns}). The integers in the set $\{a_0, a_1, a_2, \dots, a_{t-1},a_t\}$ are acquaints, and none of them has exactly one sibling.

    Conversely, assume there is a set of acquaints in which no vertex has one sibling. 
    Since we assume the network is binary,  every vertex will appear in exactly two subsets; otherwise, it would imply an in-degree greater than 2, which is not allowed in a binary network. 
    On the other hand, if a vertex appeared in exactly one subset, this would imply that it had only one sibling, which violates the assumption.
    Therefore, we have a set of a type described in \Cref{tab:dictionary-crowns}, and the network has a crown.
\end{proof}

According to the theorem above, we can use \Cref{a:count.embedded.covers} to count the number of embedded partitions by enumerating the acquaints of all integers that are inside a set of size 2, because, in the definitions of crown and fences (\Cref{tab:dictionary-crowns}), they do not contain sets of any other sizes.

\begin{algorithm}
    \caption{Count the number of support trees in a binary tree-based network.}    
    \begin{algorithmic}
        \Procedure{TraverseAcquaints}{$\tau$, $i$, $A$}
            \State {add $i$ to $A$}
            \State {mark $i$ in $\tau$ as visited}
            \For {$s \in \tau_i$ such that $s$ is not marked as visited}
                \State {\textsc{TraverseAcquaints}$(\tau, s, A)$}
            \EndFor
        \EndProcedure
        
        \Procedure{CountSupportTreesBinaryNetworks}{$\C$}
            \State $\sigma_i \gets $number of times integer $i$ appears in $\C$
            \State $\C^{=2} \gets $ all subsets of $\C$ of size 2
            \State $I \gets $ integers appearing in $\C^{=2}$ 
            \State {$\tau \gets $ table in which index $i \in I$ contains all siblings of $i$} \Comment{See \Cref{tab:dictionary}}
            \State $c \gets 0$ \Comment{Number of crowns}
            \State $f \gets 1$ \Comment{Number of fences}
            \While{$\tau$ not all $i$ are marked as visited}
                \State select $i$ not marked as visited in $\tau$
                \State $A \gets \emptyset$ 
                \State {\textsc{TraverseAcquaints}$(\tau, i, A)$} \Comment{$A$ will contain the acquaints of $i$}

                \State {$A_1 \gets $ set of integers $i$ in $A$ that have one sibling in $\tau_i$}
                \If {$|A_1| = 0$}
                    \State {$c \gets c + 1$} \Comment{\Cref{t:crowns.acquaints}}
                \EndIf
                \If {$|A_1| = 2$}
                    \State {$a,b \in A_1$}
                    \If {$\sigma_a = 1 $ and $\sigma_b = 1$}
                        \State {$f \gets f \times (|A| - 1)$} \Comment{\Cref{t:fences.acquaints}. By \Cref{t:counting.support.trees}, we have to add 1 to the length}
                    \EndIf
                \EndIf

            \EndWhile
            \State \textbf{return} $2^c \times f$
        \EndProcedure
    \end{algorithmic}
    \label{a:count.embedded.covers}
\end{algorithm}

\begin{eg}
    We saw in~\Cref{eg:embedded.partitions} that the cover for the binary tree-based network in~\Cref{f:network} has four embedded partitions, and hence the network has four support trees (shown in~\Cref{f:support-trees}).  These can be counted using~\Cref{t:counting.support.trees} as follows.  The cover $\C=1\mid 2\mid 3\mid 4,5\mid 6,8\mid 6,7\mid 7,8\mid 11,12\mid 9,13\mid 10,13\mid 14,15$ has one crown, namely $6,8\mid 6,7\mid 7,8$, and one fence $9,13\mid 10,13$, which has length 1 (a single reticulation).  Hence, the number of support trees is $2^2\times (1+1)=4$, as expected.
\end{eg}

\begin{table}[ht]
    \centering
    \begin{tabular}{lp{10cm}}
         \hline
         \textbf{Network} & \textbf{Cover}  \\ 
         \hline
         Crown & Collection of sets $a_0,a_1\mid a_1,a_2\mid\dots\mid a_{t-1},a_t\mid a_t,a_0$.\\
         Fence  & Collection of sets $a_0,a_1\mid a_1,a_2\mid\dots\mid a_{t-1},a_t\mid a_t,a_{t+1}$ with $a_0\neq a_{t+1}$ both appearing uniquely.\\
         \hline\\
    \end{tabular}
    \caption{Translation of concepts arising from counting support trees for binary tree-based networks.}
    \label{tab:dictionary-crowns}
\end{table}

\section{Tree-child networks}\label{s:tree-child}

Tree-child networks are phylogenetic networks for which every vertex has a child that is a tree vertex~\cite{cardona2008comparison}.  They satisfy a number of important properties. For instance, they have the property that every vertex is \emph{visible}.  This is a property that we describe in~\Cref{s:visible}, but first, tree-child networks turn out to have a very natural description in terms of covers, as follows.

\begin{thm}\label{t:treechild.covers}
    Tree-child networks are in bijection with expanding covers for which each set contains an integer that appears exactly once in the cover.
\end{thm}

\begin{proof}
The proof relies on the fact that the integers that appear precisely once in a cover are exactly the tree vertices.  

    Let $N$ be a tree-child network with expanding cover $\C$. Each non-leaf vertex $v$ in $N$ corresponds to a specific set $C_v$ in $\C$, whose elements label the children of $v$ in $N$.  Because $N$ is a tree-child network, each such vertex $v$ has at least one child that is a tree vertex.  The labels of the tree vertices appear precisely once in the cover, so the set $C_v$ contains at least one element that appears precisely once in the cover. This holds for every non-leaf vertex, and so for every set in $\C$, which establishes the forward direction.

    The reverse direction is also straightforward.  Suppose that every set in an expanding cover $\C$ has an element that appears precisely once in $\C$.  Since each set in the cover is the set of labels of the children of a non-leaf vertex, this implies that every non-leaf vertex has at least one child whose label appears once in the cover. In other words, it is a tree vertex.  Thus, the network corresponding to $\C$ is a tree-child network.
\end{proof}

\subsection{Visible vertices}\label{s:visible}

An important property of tree-child networks is that all of their vertices are \emph{visible}~\cite[Lemma 2]{cardona2008comparison}.  A vertex $v$ in a network is visible if there is a leaf $x$ for which every path from the root to $x$ passes through $v$.  In this section, we show how visibility can be interpreted by using covers, beginning with the definition of the \emph{backtrack} of a label in a cover.

\begin{defn}\label{d:backtrack}
    Let $\C$ be an expanding cover in labelling order and let $x$ be an element of $[m]$. Then a \emph{backtrack for $x$} is a sequence of sets 
    $S_1,\dots,S_t$ in $\C$ for which the label of a set containing $x$ is in $S_1$, and the label of $S_i$ is an element of $S_{i+1}$ for each $i=1,\dots,t-1$.  This corresponds to the output of~\Cref{a:labellingorder.backtrack}.  Let $B_\C(x)$ denote the set of all backtracks of $x$ in $\C$.
\end{defn}

\begin{algorithm} 
\caption{Backtracking algorithm}
\begin{algorithmic}
\Require {Expanding cover $\C$ in labelling order}
\Procedure{Backtrack}{$\C$, $x$}
    \State {seq $\gets$ [\quad]}
    \State {$s \gets$ a subset of $\C$ containing $x$}
    \While {label of $s$ is not $\rho$}
        \State {$s \gets$ a subset of $\C$ that contains the label of $s$ as an element}
        \State {add $s$ to seq}
    \EndWhile
    \State \textbf{return} seq
\EndProcedure
\end{algorithmic}
\label{a:labellingorder.backtrack} 
\end{algorithm}

We can characterise visibility in a network by using the backtracking algorithm. Given $x \in [m]$ and a backtrack $\beta$ for $x$, we define $L(\beta) = \{\text{label of } s \,|\, s \in \beta \}$. In this way, $L(\beta)$ contains the vertices of a path from $x$ to $\rho$ (the root),  $\bigcup\limits_{\beta\in B_C(x)} L(\beta)$ is the set of all vertices that \emph{can be} visited with a path from $x$ to $\rho$, and $\bigcap\limits_{\beta\in B_C(x)} L(\beta)$ is the set of all vertices that \textit{must be} visited on a path from $x$ to $\rho$.

\begin{thm}\label{t:single.visible.backtrack}
    Given a cover $\C$ in labelling order and $x \in [m]$, $x$ is a visible vertex in the corresponding network if and only if there exists $y \in [n]$ such that $x \in \bigcap\limits_{\beta\in B_\C(y)} L(\beta)$.
\end{thm}
\begin{proof}
    For the forward direction, assume that a vertex $x$ of a network is visible. By definition, there exists a leaf $y$ (in other words, $y \in [n])$ such that all paths from the root to $y$ pass through $x$. Since $\bigcap\limits_\beta L(\beta)$ is the set of all vertices we \textit{have to} visit from $y$ to $\rho$,  $x$ must be in this intersection.

    For the backward direction, let $y \in [n]$ and $x \in \bigcap\limits_{\beta\in B_\C(y)} L(\beta)$. Then it means that all paths from $y$ to $\rho$ contain $x$. Therefore $x$ is visible in the corresponding network.
\end{proof}

Since all $x \in \bigcap\limits_{\beta\in B_C(x)} L(\beta)$ are visible vertices and vice versa, we obtain the following corollary.

\begin{cor}\label{c:all.visible.backtrack}
    Given a cover $\C$ in labelling order, then all $x \in \bigcup\limits_{y \in [n]}\bigcap\limits_{\beta\in B_\C(y)} L(\beta)$ are visible vertices in the corresponding network and vice versa.
\end{cor}

\begin{table}[ht]
    \centering
    \begin{tabular}{lp{10cm}}
         \hline
         \textbf{Network} & \textbf{Cover}  \\ 
         \hline
         Path from node $x$ to the root & A backtrack for $x$ \\
         Visible vertex $x$  & There is a $y\in [n]$ such that $x\in\displaystyle\bigcap_{\beta\in B_\C(y)}L(\beta)$.\\
         \hline\\
    \end{tabular}
    \caption{Translation of the concepts arising in tree-child networks. }
    \label{tab:dictionary-backtrack}
\end{table}

\subsection{Support trees for tree-child networks}

\Cref{t:treechild.covers} allows us to provide an alternative proof of a result about support trees in tree-child networks, as follows.

\begin{cor}[\cite{francis2018identifiability}, Theorem 3.3]\label{c:binary.treechild.support}
    A binary tree-child network with $k$ reticulations has $2^k$ support trees.
\end{cor}
\begin{proof}
     Since each set in the cover for a tree-child network has a uniquely appearing element, there are no sets containing only reticulations (i.e. no singletons with elements that appear elsewhere, and no pairs in which both elements are repeated). Using the categories above, all sets in such a cover are from Categories~\eqref{singletons.once}, \eqref{pairs.once}, or \eqref{pairs.mixed}.

     As a consequence, there are no crowns, which require sets with two reticulations, and each fence can only have length 1, being of the form $a,b\mid b,c$, and containing only one reticulation ($b$ in this case).  Furthermore, each repeated integer in the cover (i.e., each reticulation) is in a fence, since it must be part of a pair with a uniquely appearing element (a tree vertex).  Thus, the number of fences is the number of reticulations, and each fence has length 1.
     Therefore, by~\Cref{t:counting.support.trees}, there are $2^k$ support trees.
\end{proof}

Corollary~\ref{c:binary.treechild.support} also follows immediately by combining both parts of the following result.
\begin{thm}
\mbox{}
\begin{itemize}
\item[(i)]     The number of spanning trees in a phylogenetic network is the product of all the in-degrees of the reticulation vertices.
\item[(ii)] A network is a tree-child network if and only if every spanning tree is also a support tree.
\end{itemize}
\end{thm}
\begin{proof}
Part (i):     A reticulation vertex $x$ is an integer contained in $k > 1$ subsets of $\C$ (\Cref{tab:dictionary}) and a spanning tree is an embedded partition (\Cref{tab:dictionary-spanning.tree}). Thus, to obtain an embedded partition from a cover, we have to remove $k - 1$ instances of $x$ from $\C$. This can be done in ${k \choose k-1} = k$ different ways, and each choice is independent of the others. Since $k$ is also the in-degree for vertex $x$, it follows than the number of embedded partitions (spanning trees) is $\prod_x \textit{in-degree}(x)$, where $x$ is a reticulation vertex. If $x$ is a tree vertex, then in-degree$(x) = 1$ and, therefore, it does not contribute to the product. 
 
Part (ii): By \Cref{t:treechild.covers},  every subset of a tree-child cover has at least one element that is not present in any other subset.  This implies that every embedding partition must contain at least one element for each subset; hence, it has the same size as $|\C|$.
 
To show the forward direction, suppose that $N$ is not a tree-child network. We will show that there must be a spanning tree for $N$ that is not a support tree. If $N$ is not tree-child, then it has at least one vertex that is not visible. Let $v$ be a non-visible vertex that is maximally distant from the root, so that all vertices descended from $v$ are visible. If we delete each arc out of $v$, then there is a path from the root to each vertex,  so $N$ has a spanning tree $T$. However, in this tree, $T$ has $v$ as a leaf. The tree $T$ is therefore a spanning tree  of $N$ and not all its leaves are in $X$, so $T$ is not a support tree.
\end{proof}

\section{Normal networks}\label{s:normal}

Normal networks are a subclass of the tree-child networks, with the added constraint that they contain no ``shortcuts"~\cite{willson2010properties}.  A \emph{shortcut} is an edge $(u,v)$ for which there is an alternative directed path from $u$ to $v$ in the network.

To capture this information in terms of covers, we need a way to record paths in that context.  This motivated the definition of backtrack (\Cref{d:backtrack}), which requires the labelling order that was defined in~\Cref{s:prelim}.
The backtrack algorithm identifies a path from the vertex labelled $x$ back to the root, expressing the path in terms of a sequence of sets in the cover.  The edges between the parent vertices that correspond with these sets defines the path.

\begin{eg}\label{eg:backtrack}
    Recall the cover $\C=1\mid 2\mid 3\mid 4,5\mid 6,8\mid 6,7\mid 7,8\mid 11,12\mid 9,13\mid 10,13\mid 14,15$ from~\Cref{eg:embedded.partitions} for the network in~\Cref{f:network}. This cover has the labelling order
    \[\C=\{1\}_6,\{2\}_7,\{3\}_8,\{4,5\}_9,\{6,7\}_{10},\{6,8\}_{11},\{7,8\}_{12},\{11,12\}_{13},\{9,13\}_{14},\{10,13\}_{15},\{14,15\}_{\rho}.\]
A backtrack for $x=3$ starts with a subset containing $8$ (with the label of $\{3\}$ in the labelling order).  There are two choices; suppose we pick $\{7,8\}$.  The label of $\{7,8\}$ is 12, so now we must find a set containing $12$. There is only one, so we add $\{11,12\}$ to the backtrack sequence.  $\{11,12\}$ has label 13, so we look for a set containing $13$ and choose one of the two options, say $\{10,13\}$.  This has label 15 in the order, so we look for a set containing $15.$  There is one, namely $\{14,15\}$, and its label is $\rho$, which means we terminate the algorithm and output the backtrack sequence
    \[\{7,8\},\{11,12\},\{10,13\},\{14,15\}.\]
    Note, each such backtrack defines a path from 3 to the root $\rho$; in this case, $3\to 8\to 12\to 13\to 15\to\rho$.
\end{eg}

\begin{thm}\label{t:labellingorder.shortcuts}
    Let $N$ be a phylogenetic network with expanding cover $\C$, in labelling order.  Then $N$ has a shortcut if and only if there is a backtrack for an $x\in[m]$ that includes a subset containing $x$.
\end{thm}
\begin{proof}
    Suppose $N$ has a shortcut.  Then there is a vertex $x$ with a non-trivial path from some vertex $v$ to $x$, and there is also an edge $(v,x)$.  The existence of a non-trivial path from $v$ to $x$ means that the cover has a non-trivial backtrack from $x$, which includes the children of $v$ as a set. However, $x$ is also a child of $v$, so $x$ is in a set in the backtrack. 

    Conversely, suppose that the cover contains a backtrack for $x$ that includes a set $S$ containing $x$.
    Let $v$ be the label of the parent of $S$.  Then $x$ is a child of $v$, meaning there is an edge $(v,x)$ in $N$. However, the backtrack provides a non-trivial path in $N$ from $v$ to $x$ through $S$.  That is, $N$ contains a shortcut.
\end{proof}

\begin{cor}\label{c:labellingorder.normal}
    Let $\C$ be a cover in labelling order for a tree-child network. Then $\C$ is a cover for a normal network if and only if, for all $x \in [m]$, no backtrack for $x$ has a subset that contains $x$.
\end{cor}

Without loss of generality, in \Cref{t:labellingorder.shortcuts} and \Cref{c:labellingorder.normal}, we can assume that $x$ is a reticulation vertex (i.e., a value in $[m]$ that is contained in more that one subset of $\C$), since, by definition, reticulations have in-degree greater than one and thus are the only vertices that can have shortcuts. 

Using \Cref{t:labellingorder.shortcuts}, we can construct an algorithm that removes all the shortcuts from a cover.
This implies that, given a tree=child network, we can transform it to a normal network by removing all the shortcuts via \Cref{a:remove.shortcuts}.

\begin{algorithm}[ht]
    \caption{Remove all shortcuts from a cover}
    \begin{algorithmic}
    \Require {Expanding cover $\C$ in labelling order}
        \Procedure{RemoveShortcuts}{$\C$}
        \State Compute all backtracks for all reticulation vertices
        \For{backtrack $\beta$ for reticulation vertex $x$}
            \For{$s \in \beta$}
                \If{$x \in s$}
                    \State Remove $x$ from $s$
                \EndIf
            \EndFor
        \EndFor
        \EndProcedure
    \end{algorithmic}
    \label{a:remove.shortcuts}
\end{algorithm}

\begin{table}[ht]
    \centering
    \begin{tabular}{lp{10cm}}
         \hline
         \textbf{Network} & \textbf{Cover}  \\ 
         \hline
         Shortcut to $x$  & A backtrack of $x$ that includes a set containing $x$.\\
         \hline\\
    \end{tabular}
    \caption{A translation of a shortcut into a feature of the corresponding expanding cover. }
    \label{tab:shortcut}
\end{table}

\section{Tree-sibling networks} \label{sec:treesibling}

Tree-sibling networks are also amenable to a description in terms of covers.

\begin{defn}[\cite{cardona2008distance}]
    A tree-sibling network is a network in which every reticulation vertex is a sibling of a tree vertex.
\end{defn}

\begin{thm}\label{t:treesibling.covers}
    Tree-siblings networks are in bijection with those expanding covers 
    for which every repeated integer lies in at least one set with an integer that appears only once.
\end{thm}
\begin{proof}
    The statement is a direct translation of the definition of tree-sibling into the language of covers, {according to \Cref{tab:dictionary}}.
    Reticulation vertices are those that appear more than once in the cover, and vertices are siblings when they appear in the same set in the cover. 
\end{proof}

We have already seen a characterisation of tree-child networks using covers in~\Cref{t:treechild.covers}.  Covers for tree-child networks are those for which every set has a uniquely appearing element.   However, there is a close connection between tree-child and tree-sibling networks, which can be captured in a cover description for tree-child networks, as follows.

\begin{thm}\label{t:treesichild.from.treesibling}
    Tree-child networks are in bijection with expanding covers for which, for every
    repeated element $k$ in $\C$, {every} subset
    containing $k$ also contains an integer that appears only once.
\end{thm}
\begin{proof}
    We will prove that this statement is equivalent to \Cref{t:treechild.covers}.
    
   For the forward direction, suppose that a cover satisfies the condition in~\Cref{t:treechild.covers}. If every subset contains a uniquely occurring integer, then all subsets that contain a reticulation will do also do so. 

    For the backward direction, by assumption, every subset that contains a reticulation vertex has an integer that is not contained in another subset. This implies that all other subsets do not contain a reticulation vertex, and therefore, it contains a tree vertex that is not contained in any other subset (\Cref{tab:dictionary}).
\end{proof}
In other words, tree-child networks are networks in which every parent of a reticulation vertex has a tree-vertex as a child. Therefore, we recover the well-known fact that all tree-child networks are tree-sibling networks.

\section{Orchard networks}
\label{s:orchard}

Orchard networks are non-degenerate phylogenetic networks defined by the property that they can be reduced to a trivial network (a single vertex) by a series of cherry or reticulated cherry reductions~\cite{erdos2019class,janssen2021cherry,van2022orchard}.  In the present paper, we will restrict our attention to \emph{binary} orchard networks.

A \emph{cherry} is a pair of leaves that are siblings; a \emph{reticulated cherry} is a pair of leaves, one of which has a reticulate parent and the other is the sibling of that reticulate parent.
Cherry reduction involves replacing the cherry with a single vertex.  Reticulated cherry reduction involves deleting the arc between the parents of the two leaves and then suppressing degree-2 vertices.   
By a theorem of~\cite{erdos2019class,janssen2021cherry}, for orchard networks, the order in which these are performed is not important. 

To translate this definition into covers, we need to first characterise cherries and reticulated cherries as they are manifested in covers, and then describe the action of such reductions in terms of the cover.  The first of these requirements is routine; the second, not, as it requires us to augment the cover with its set of leaves.  We will describe a test for orchard that reduces an expanding cover to a trivial cover but, along the way, passes through covers that are not expanding.

In covers, a cherry is given by a set consisting of two elements of $[n]$ (the leaves), whereas a reticulated cherry is given by a singleton subset of $[n]$ appearing in position $j$ in the labelling order, and a pair $\{n+j,i\}$ where $i\in [n]$ (summarized in~\Cref{tab:orchard}).  An example is shown in~\Cref{f:cherries}.

\begin{figure}[ht]
    \centering
    \includegraphics{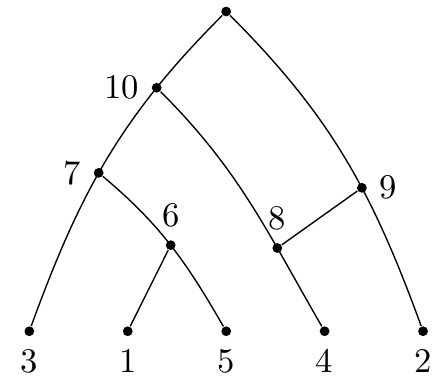}
    \caption{A network with a cherry and a reticulated cherry. The cover, in labelling order, is $(\{1,5\},\{3,6\},\{4\},\{2,8\},\{7,8\},\{9,10\})$.  The cherry can be identified in the cover as a pair of integers that are a subset of the leafset $[5]$.  In this cover, a cherry is $\{1,5\}$.  The reticulated cherry is identified in the cover as a pair of sets: one is a singleton subset of the leafset, in position $j$;  the other is the pair $\{n+j,i\}$ with $i$ in the leafset.  In this cover, there is a reticulated cherry consisting of the singleton $\{4\}$ (contained in the leafset), which appears in position 3 in the labelling order, and the pair $\{2,8\}$, noting that $8=n+3$ and $2$ is in the leafset.}
    \label{f:cherries}
\end{figure}

\subsection{The cherry reduction process via covers}

The cherry reduction test for orchard networks can be defined efficiently using covers by keeping track of the changing set of leaf labels $\LL$ within the algorithm, as follows.
Identifying a cherry or reticulated cherry in a cover can be done using the translations given in~\Cref{tab:orchard}.  The process in~\Cref{alg:orchard.test} chooses to reduce a cherry first, if there is one, as it involves fewer checks.

\begin{algorithm}[ht]
    \caption{Test whether the expanding cover $\C$ corresponds to an orchard network}\label{alg:orchard.test}
    \begin{algorithmic}
    \Require {Expanding cover $\C$ in labelling order}
    \Procedure{IsOrchard}{$\C$}
        \State $n\gets ||\C||-|\C|+1$
        \State $\LL\gets [n]$
        \State $reduced \gets true$
        \While{$reduced = true$ \textbf{and} $|\C|>0$}
            \State $reduced \gets false$
            \If{there is a set of form $\{a,b\}_j\in\C$ with $a,b\in\LL$} 
                \State $\C\gets\C\setminus\{a,b\}$ \Comment{Cherry reduction}
                \State $\LL\gets(\LL\setminus\{a,b\})\cup\{j\}$
                \State $reduced \gets true$
            \Else
            \If{there is a set of form $\{a\}_j$ in $\C$ and a set of form $\{j,b\}\in\C$, with $a,b\in\LL$}    
                    \State $\C\gets\C\setminus\{\{a\},\{j,b\}\}$ \Comment{Reticulated cherry reduction}
                    \State $\LL\gets (\LL\setminus\{a,b\})\cup\{j,k\}$
                    \State $reduced \gets true$
                
            \EndIf
            \EndIf
        \EndWhile
        \If{$\C=\emptyset$}\label{alg.ref:go.to.output}
            \State \textbf{return} ``$\C$ is orchard''
        \Else 
            \State \textbf{return} ``$\C$ is not orchard''
        \EndIf
    \EndProcedure
    \end{algorithmic}
    \label{a:test.orchard}
\end{algorithm}

In general, the set $\C$ that is redefined during~\Cref{alg:orchard.test} may not be an expanding cover, but these processes do nevertheless model the network cherry and reticulated cherry reduction steps, applied to a labelled network.

\begin{thm}
    \Cref{alg:orchard.test} determines whether the network from the expanding cover $\C$ is orchard.
\end{thm}
 
\begin{proof}
    A network is orchard, by definition, if and only if it can be reduced to a trivial network by cherry or reticulated cherry reductions. According to a result of~\cite{erdos2019class,janssen2021cherry}, the order of such reductions is not important.  The procedures in~\Cref{alg:orchard.test} exactly reflect the effect on the cover of these operations on the network, as can be seen in~\Cref{f:cherry.reductions}.
\end{proof}

\begin{figure}[ht]
    \centering
    \includegraphics[width=\textwidth]{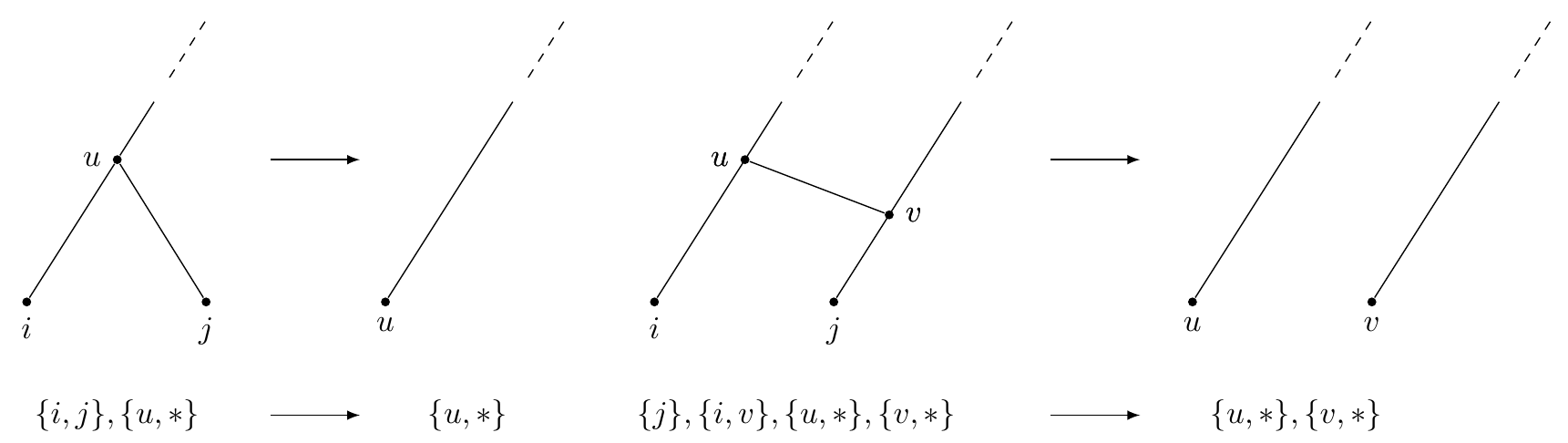}
    \caption{Cherry (left) and reticulated cherry (right) reductions and their effects on the covers.  Here, $\ast$ represents other sibling vertices, which could be an empty set. }
    \label{f:cherry.reductions}
\end{figure}

\begin{eg}\label{eg:cherry-reduction}
    The cherry reduction process in~\Cref{alg:orchard.test}, applied to the cover for the network in~\Cref{f:cherries}, proceeds as described in~\Cref{tab:cherry.reduction}.

\begin{table}[ht]
    \centering
\begin{tabular}{|l|p{11cm}|c|}
\hline
\multicolumn{3}{|p{\textwidth}|}{Cherry reduction of the cover $\C$ in~\Cref{eg:cherry-reduction}, following~\Cref{alg:orchard.test}.}\\ \hline
\multicolumn{3}{|p{\textwidth}|}{$\C=\left\{\{1,5\}_6,\{3,6\}_7,\{4\}_8,\{2,8\}_9,\{7,8\}_{10},\{9,10\}_\rho\right\}$;}\\
\multicolumn{3}{|p{14cm}|}{$\cL=\{1,2,3,4,5\}$.}\\
\hline
\multirow{2}{*}{1} & \multicolumn{2}{|p{15cm}|}{
$\C$ contains the cherry $\{1,5\}_6$ (and the reticulated cherry, $\{4\}_8,\{2,8\}_9$).  We reduce the cherry.}  \\
\cline{2-3}
    & 
        $\begin{aligned}[b]
         \C_1&=\C\setminus\{\{1,5\}_6\}\\
         &=\left\{\{3,6\}_7,\{4\}_8,\{2,8\}_9,\{7,8\}_{10},\{9,10\}_\rho\right\} \\
         \cL_1&=(\cL\setminus\{1,5\}) \cup\{6\}\\
         &=\{2,3,4,6\}
        \end{aligned}$ 
    & {\includegraphics[width=2.5cm]{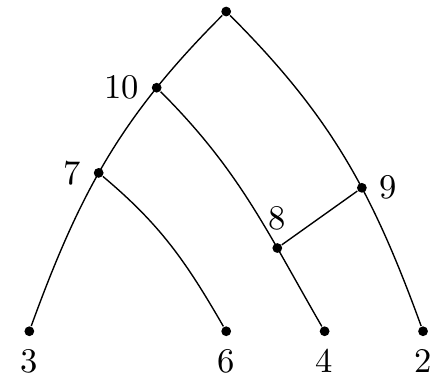}}\\
\hline
\multirow{2}{*}{2} & \multicolumn{2}{|p{15cm}|}{
$\C_1$ contains cherry $\{3,6\}_7$ (and the reticulated cherry $\{4\}_8,\{2,8\}_9$). We reduce the cherry.} \\
\cline{2-3}
    & 
    $\begin{aligned}[b]
    \C_2&=\C_1\setminus\{\{3,6\}_7\}\\  
        &=\left\{\{4\}_8,\{2,8\}_9,\{7,8\}_{10},\{9,10\}_\rho\right\}.\\ 
    \cL_2&=(\cL_1\setminus\{3,6\}) \cup\{7\}\\
        &=\{2,4,7\}.    
    \end{aligned}$
    &     \includegraphics[width=2.5cm]{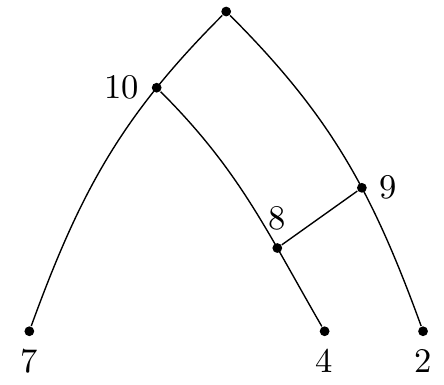}\\    
\hline
\multirow{2}{*}{3} & \multicolumn{2}{|p{15cm}|}{
$\C_2$ contains no cherry, but contains the reticulated cherry $\{4\}_8,\{2,8\}_9$, which we reduce.  } \\
\cline{2-3}
    & 
    $\begin{aligned}[b]
    \C_3&=\C_2\setminus \{\{4\}_8,\{2,8\}_9\}\\
        &=\left\{\{7,8\}_{10},\{9,10\}_\rho\right\}.\\ 
    \cL_3&=(\cL_2\setminus\{2,4\})\cup\{8,9\}\\
        &=\{7,8,9\}.
    \end{aligned}$
    &     \includegraphics[width=2.5cm]{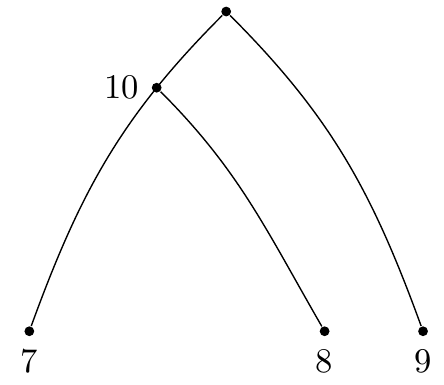}\\    
\hline
\multirow{2}{*}{4} & \multicolumn{2}{|p{15cm}|}{
$\C_3$ contains the cherry $\{7,8\}_{10}$, which we reduce.} \\
\cline{2-3}
    & 
    $\begin{aligned}[b]
    \C_4&=\C_3\setminus \{7,8\}_{10}\\
        &=\left\{\{9,10\}_\rho\right\}.\\
    \cL_4&=(\cL_3\setminus \{7,8\})\cup\{10\}\\
        &=\{9,10\}.
    \end{aligned}$
    &     \includegraphics[width=2.5cm]{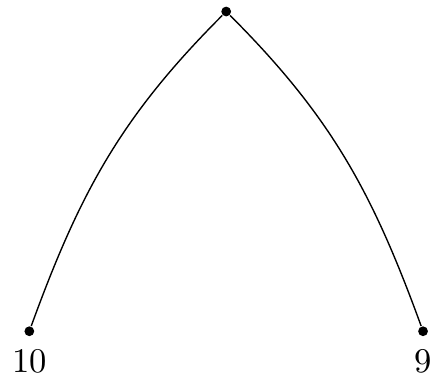}\\    
\hline
\multirow{2}{*}{5} & \multicolumn{2}{|p{15cm}|}{
$\C_4$ contains (only) the cherry $\{9,10\}_\rho$, which we reduce.} \\
\cline{2-3}
    & 
$\C_5=\C_4\setminus \{9,10\}_\rho = \emptyset$, which means the algorithm ends.
    &      \\
\hline
\end{tabular}    
\caption{Cherry reduction algorithm acting on the cover in~\Cref{eg:cherry-reduction} via~\Cref{alg:orchard.test}, with the effects of reduction on the network shown at right (for illustration only). }
\label{tab:cherry.reduction}
\end{table}

\end{eg}

\begin{table}[ht]
    \centering
    \begin{tabular}{lp{10cm}}
         \hline
         \textbf{Network} & \textbf{Cover}  \\ 
         \hline
         Cherry & A set consisting of two elements of $[n]$ \\
         Reticulated cherry & A singleton subset of $[n]$ appearing in position $j$ in the labelling order, and a pair $\{n+j,i\}$ where $i\in [n]$ \\
         \hline\\
    \end{tabular}
    \caption{A translation of features that are relevant to orchard networks into features of the corresponding expanding cover. }
    \label{tab:orchard}
\end{table}

\section{A new class of network detected through the lens of covers}
\label{s:new.classes}
\label{s:spinal}

We have used covers to describe several classes of phylogenetic network.  However, the encoding into covers also creates the opportunity to define new classes of network that correspond to particular features of covers.  Such classes might currently have little direct utility for application to phylogenetics, but they may have an indirect value in that algorithms and methods using covers may involve such classes in passing.
We introduce one such class as an example of this opportunity.

Recall that the definition of an expanding cover has two criteria (\Cref{d:expanding.cover}).  The first is that elements of the leafset $[n]$ are not repeated, and the second ensures that the labelling algorithm is well-defined by requiring at least $i$ subsets of $[n+i-1]$ to be in the cover.

If a cover contains \emph{exactly} $i$ subsets of $[n+i-1]$, it has a strong consequence for the network, as follows.  We define a \emph{spine} in a network to be a path from a leaf to the root that traverses all non-leaf vertices, and we call a network \emph{spinal} if it has a spine.

\begin{thm}\label{t:spinal.characterisation}
    A network is spinal if and only if its cover $\C$ has exactly $i$ subsets of $[n + i - 1]$, for each $i=1,\dots,|\C|$.
\end{thm}
\begin{proof}
We prove the reverse direction first.  Suppose that the cover $\C$ has exactly $i$ subsets of $[n + i - 1]$, for each $i=1,\dots,|\C|$, and consider its labelling order.
The first set in the labelling order is the unique set that is contained in $[n]$, and its label is $n+1$. 
For each $i=1,\dots,|\C|-1$, the $i$th set in the labelling order is contained in $[n+i-1]$, according to the expanding property, but is not contained in $[n+i-2]$, according to our assumption about $\C$. Therefore, it must contain the integer $n+i-1$.  The $i$th set in the labelling order has label $n+i$, which means that $n+i$ is a parent of $n+i-1$.  
Since this holds for each $i>1$, this determines a path from a leaf (labelled by an element of the first set in the labelling order) through every vertex with label $n+1$ to $n+|\C|-1$, and the last set, containing $||\C||=n+|\C|-1$, has the root as parent.  Thus the network is spinal.

We now prove the forward direction.  Suppose that $N$ is spinal with cover $\C$.  Being spinal means that $N$ has a path of length $|\C|$ from a leaf to the root.  This means that there is a backtrack of a leaf that has length $|\C|-1$.  That is, a sequence of sets from the cover such that the label of one set (from the labelling order on $\C$) is an element of the next set in the backtrack sequence. 
Because the label of a set in the cover is strictly greater than all the elements of the set, the maximal elements of the sets in a backtrack are {strictly increasing}.

Now consider the backtrack arising from the spine (the path from a leaf to the root traversing all non-leaf vertices).  
The leaf at the base of the spine must be in a set contained in $[n]$; otherwise, there would be no path from it to the vertex labelled $n+1$. Therefore, the first set in the backtrack contains $n+1$ as its maximal element because that is the parent label for the set containing the initial leaf. 
The spine has $|\C|-1$ vertices in it, including the initial leaf, because it includes all except $n-1$ of the vertices in the network (the network has $||\C||+1=|\C|+n$ vertices in total).  Therefore, the backtrack for the initial leaf has $|\C|-1$ sets.   
The maximal elements of these $|\C|-1$ sets are strictly increasing, and run from $n+1$ to $m=||\C||=|\C|+n-1=n+(|\C|-1)$. This forces each set in the backtrack to have a distinct maximal element.  Put together with the set containing the initial leaf, which is a subset of $[n]$, this means that there are exactly $i$ subsets of $[n+i-1]$, for each $i=1,\dots,|\C|$, as required.
\end{proof}

In the light of \Cref{t:spinal.characterisation}, we say that a cover $\C$ is \emph{spinal} if it contains {exactly} $i$ subsets of $[n+i-1]$, for each $i=1,\ldots, |\C|$. 
An example of a spinal network is shown in~\Cref{f:spinal}.
Spinal networks have some non-trivial intersections with other classes; for example, the spinal network
    $1 \mid 2 \mid 2,3 \mid 3,4$
is not a tree-child, tree-sibling, or orchard network. It can, however, be shown that the class of spinal networks lies within the intersection of the labellable and tree-based classes of networks.

\begin{figure}[ht]
    \centering
    \includegraphics[scale=0.9]{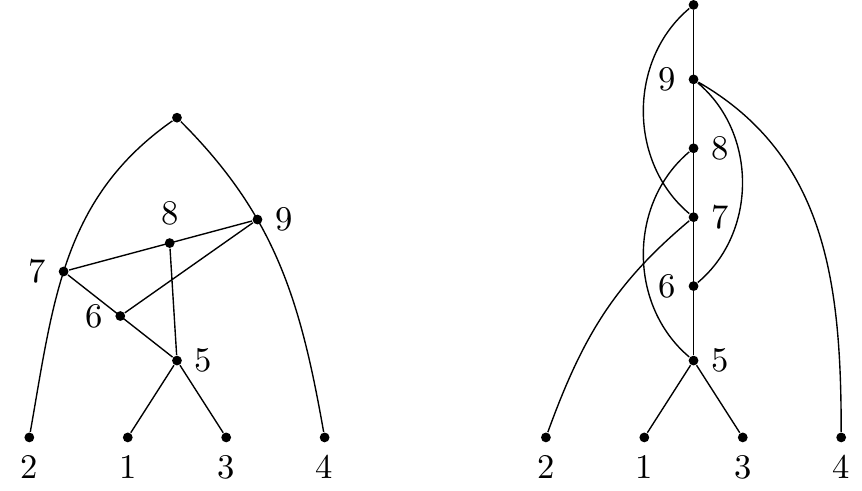}
    \caption{A spinal network with cover $1,3\mid 5\mid 2,6\mid 5,7\mid 4,6,8\mid 7,9$. Note that $n=4$ and the cover has one set in $[4]$, two in $[5]$, three in $[6]$, four in $[7]$, five in $[8]$, and six in $[9]$. There is a path from the elements of the set that is in $[4]$, namely 1 and 3, to the root, and this path traverses every non-leaf vertex.  Observe that this path is in labelling sequence.  The spine is particularly clear when the network is drawn as shown on the right. }
    \label{f:spinal}
\end{figure}

\begin{table}[ht]
    \centering
    \begin{tabular}{lp{10cm}}
         \hline
         \textbf{Network} & \textbf{Cover}  \\ 
         \hline
         Spine  & Exactly $i$ subsets of $[n+i-1]$, for each $i$\\
         \hline\\
    \end{tabular}
    \caption{A translation of the feature of spinal networks into a feature of the corresponding expanding cover. }
    \label{tab:spinal}
\end{table}

\section{Discussion}
\label{s:discussion}

Sometimes a relatively small shift in perspective can open up new possibilities in surprising ways.  What seems like a fairly straightforward idea in a paper by Diaconis and Holmes (the idea that rooted binary phylogenetic trees correspond to perfect matchings~\cite{diaconis1998matchings}), itself building on an elegant but simple way to label internal vertices~\cite{erdos1989applications}, was loosened slightly to yield a correspondence between phylogenetic forests and all partitions of finite sets, as well as a raft of interesting questions in semigroup theory~\cite{francis2022brauer}.  This subtle twist of an idea, like something from a Philip Pullman novel~\cite{pullman2015subtle}, seems to have opened up further opportunities that, with a further gentle twist, have opened a new canvas on which to draw phylogenetic networks~\cite{francis2023encoding}. 
Capturing the features that define different network classes  on this canvas provided the underlying motivation for this paper.

Many core features discussed in the context of networks, such as reticulations, paths, cherries, siblings, and so on, have been translated into the language of covers; a summary is given in~\Cref{tab:final.dictionary}.  These translations of features have been necessary for characterising several important classes of phylogenetic network in the language of covers.  This includes some of the most prominent classes, including normal, tree-child, tree-sibling, orchard, and tree-based networks (relationships among the classes, determined by properties of their covers, are represented in~\Cref{fig:diagram}). However there are many classes, each of which is important for its own reasons, and this list is not complete.  Some classes that have been omitted in the present paper might be difficult to define with covers (for instance, level-$k$ networks or HGT networks), whereas others might just be a matter of following through with the first steps we have taken here (for example, reticulation-visible networks, and non-binary orchard networks).

\begin{table}[ht]
    \centering
    \begin{tabular}{p{3.5cm}p{11.5cm}}
         \hline
         \textbf{Network $N$} & \textbf{Cover $\C$}  \\ 
         \hline\\[-5pt]
         Non-root vertex & An integer in $[m]$ \\
         Leaf & An integer in $[n]$ \\
         Tree vertex & An integer contained in just one subset \\
         Reticulation vertex & An integer contained in more than one subset \\
         In-degree of $x$ & The number of subsets that contain $x$ \\
         Out-degree of $x$ & Size of the subset with label $x$ in the labelling order \\
         Parents of $x$ & All the subsets that contain $x$ \\
         Siblings of $x$ & All the other integers contained in the subsets that contain $x$ \\
         Children of $x$ & The subset with label $x$ in the labelling order\\
         Spanning tree & A partition embedded in $\C$ \\
         Support tree & A full embedding of a partition in $\C$\\
         Crown & Collection of sets $a_0,a_1\mid a_1,a_2\mid\dots\mid a_{t-1},a_t\mid a_t,a_0$.\\
         Fence  & Collection of sets $a_0,a_1\mid a_1,a_2\mid\dots\mid a_t,a_{t+1}$, with $a_0\neq a_{t+1}$ both unique.\\
         Path from $x$ to $\rho$ & A backtrack for $x$ \\
         Visible  $x$  & There is a $y\in [n]$ such that $x\in\bigcap_{\beta\in B_\C(y)}L(\beta)$.\\
         Shortcut to $x$  & A backtrack of $x$ that includes a set containing $x$.\\
         Cherry & A set consisting of two elements of $[n]$ \\
         Reticulated cherry & A subset $\{a\}_k$ of $[n]$, and a pair $\{k,b\}$ with $b\in [n]$ \\
         Spine  & Exactly $i$ subsets of $[n+i-1]$, for each $i$\\
         \hline\\
    \end{tabular}
    \caption{A translation of the features of a network with $n$ leaves and $m$ non-root vertices, into the features of the corresponding expanding cover.  
}
    \label{tab:final.dictionary}
\end{table}

Defining a language is not the goal, however, despite it being a necessary step.  The goal is to be able to efficiently work with phylogenetic networks --- computationally, algorithmically, and mathematically --- in order to establish robust methods of inference for networks that will eventually be of practical use for biological researchers.  
To that end, encoding various classes of phylogenetic networks in terms of expanding covers provides an opportunity to make computation more effective and allow their structure to be seen more clearly.

\section{Data Availability}
Data sharing is not applicable to this article as no datasets were generated or analysed during the current study.

\bibliographystyle{plain}

\begin{figure}[ht]
    \centering
    \includegraphics[width =.8 \textwidth]{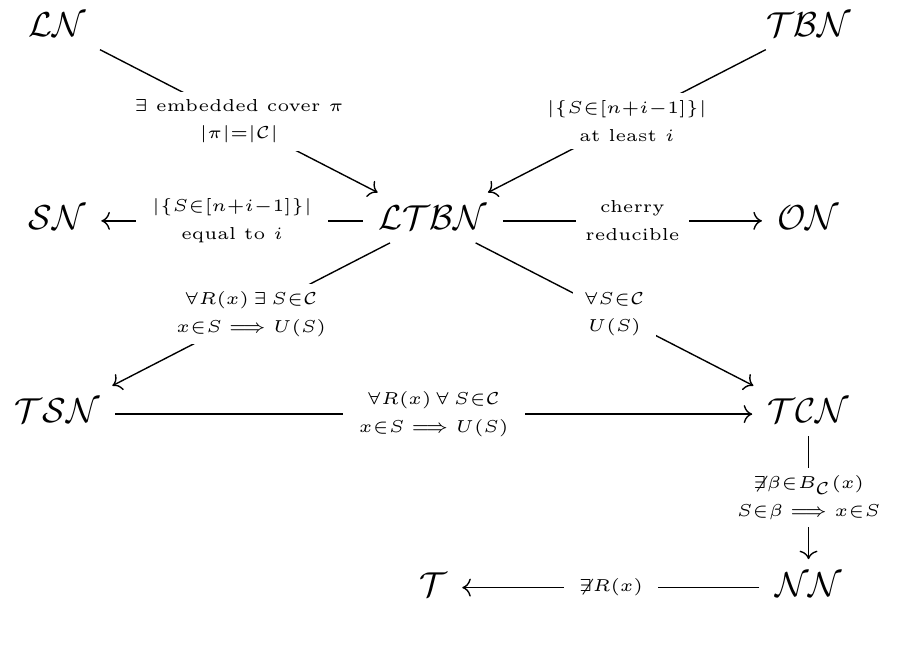}
    \caption{A diagram showing the hierarchy of networks. The nodes are classes of networks, the arrows represent inclusion and the labels indicate which axiom we add to obtain that class. $R(x)$ means ``$x$ is repeated in more than one subset'' and $U(S)$ means ``$S$ contains an unique integer''. The class labels are (from top to bottom): $\mathcal{LN}$ (labellable networks), $\mathcal{TBN}$ (tree-based networks), $\mathcal{LTBN}$ (labellable tree-based networks), $\mathcal{SN}$ (spinal networks), $\mathcal{ON}$ (orchard networks), $\mathcal{TSN}$ (tree-sibling networks), $\mathcal{TCN}$ (tree-child networks), $\mathcal{NN}$ (normal networks), and $\mathcal{T}$ (trees).}
    \label{fig:diagram}
\end{figure}

\end{document}